\documentclass[12pt]{article}
\usepackage[dvipsnames]{xcolor}
\usepackage{braket,amsfonts,amsmath,amssymb,amsthm}
\usepackage{mathrsfs,url,hyperref,tikz,graphicx}
\usepackage{comment}

\title{Creation Rate of Dirac Particles at a Point Source}

\author{
Joscha Henheik\footnote{ Institute of Science and Technology Austria (IST Austria), Am Campus 1, 3400 Klosterneuburg, Austria, E-mail: joscha.henheik@ist.ac.at}~~and
Roderich Tumulka\footnote{Mathematisches Institut,
	Eberhard-Karls-Universit\"at T\"ubingen, Auf der Morgenstelle 10, 72076
	T\"ubingen, Germany, E-mail:
     roderich.tumulka@uni-tuebingen.de}
}
\date{November 29, 2022}

\newcommand{\I}{\mathrm{i}}
\newcommand{\E}{\mathrm{e}}
\newcommand{\D}{\mathrm{d}}
\newcommand{\Hilbert}{\mathscr{H}}
\newcommand{\Kilbert}{\mathscr{K}}
\newcommand{\sA}{\mathscr{A}}

\newcommand{\Q}{\mathcal{Q}}
\renewcommand{\Re}{\mathrm{Re}}
\renewcommand{\Im}{\mathrm{Im}}

\newcommand{\RRR}{\mathbb{R}}
\newcommand{\CCC}{\mathbb{C}}
\newcommand{\SSS}{\mathbb{S}}

\newcommand{\nr}{{\mathrm{nr}}} 

\newcommand{\ve}{\boldsymbol{e}}
\newcommand{\vj}{\boldsymbol{j}}

\newcommand{\vx}{\boldsymbol{x}}

\newcommand{\vJ}{\boldsymbol{J}}
\newcommand{\vL}{\boldsymbol{L}}
\newcommand{\vQ}{\boldsymbol{Q}}
\newcommand{\vS}{\boldsymbol{S}}
\newcommand{\valpha}{\boldsymbol{\alpha}}
\newcommand{\vsigma}{\boldsymbol{\sigma}}
\newcommand{\vomega}{\boldsymbol{\omega}}
\newcommand{\vzero}{\boldsymbol{0}}

\newtheorem{prop}{Proposition}
\newtheorem{lem}{Lemma}
\theoremstyle{definition}\newtheorem{rmk}{Remark}
\newcommand{\be}{\begin{equation}}
\newcommand{\ee}{\end{equation}}

\begin{document}
\maketitle
\begin{abstract}
Only recently has it been possible to construct a self-adjoint Hamiltonian that involves the creation of Dirac particles at a point source in 3d space. Its definition makes use of an interior-boundary condition. Here, we develop for this Hamiltonian a corresponding theory of the Bohmian configuration. That is, we construct a Markov jump process $(Q_t)_{t\in\mathbb{R}}$ in the configuration space of a variable number of particles that is $|\psi_t|^2$-distributed at every time $t$ and follows Bohmian trajectories between the jumps. The jumps correspond to particle creation or annihilation events and occur either to or from a configuration with a particle located at the source. The process is the natural analog of Bell's jump process, and a central piece in its construction is the determination of the rate of particle creation. The construction requires an analysis of the asymptotic behavior of the Bohmian trajectories near the source. We find that the particle reaches the source with radial speed 0, but orbits around the source infinitely many times in finite time before absorption (or after emission).

\medskip

  \noindent 
PACS: 
11.10.Ef; 	
03.70.+k; 	
11.10.Gh. 	
  Key words: 
  ultraviolet infinity;
  Bohmian mechanics;
  regularization of quantum field theory;
  probability current.
\end{abstract}

\newpage

\section{Introduction}

It is notoriously difficult to construct quantum Hamiltonians with particle creation and annihilation at a point source. Sometimes, such Hamiltonians can be obtained through renormalization \cite{Nel64,Der03}. A more recent approach is based on interior-boundary conditions (IBCs) \cite{TT15a,TT15b}, which are mathematically related to point interactions \cite{AGHKH88,BP35}. Here, we are concerned with a particular family of self-adjoint Hamiltonians $H$ that we constructed in \cite{HT20} using IBCs.

Another ingredient in this work is Bell's jump process \cite{Bell86,DGTZ04}, which is an extension of Bohmian mechanics \cite{Bohm52,DT09,DGZ13} to quantum theories with particle creation and annihilation. These processes have been developed for theories on a lattice \cite{Bell86}, with UV cut-off \cite{DGTZ04}, and with IBCs \cite{bohmibc}. However, the processes in \cite{bohmibc} were devised for non-relativistic Schr\"odinger operators (based on the Laplacian operator) or codimension-1 boundaries (such as a surface in $\RRR^3$), whereas our $H$ is based on the Dirac operator and involves a codimension-3 boundary (corresponding to a point source in $\RRR^3$). Here, we construct an analog of Bell's jump process for $H$; its construction is somewhat more involved than the cases analyzed in \cite{bohmibc}, and it has some curious features that we report below and that are absent in the non-relativistic case.

The Hamiltonian $H$ is devised for a model of creation and annihilation of Dirac particles in 3 space dimensions by a point source fixed at the origin $\vzero\in\RRR^3$. For simplicity, our Hilbert space $\Hilbert$ is a mini-Fock space with only two sectors, corresponding to 0 or 1 particles,
\be\label{Hilbertdef}
\Hilbert =\Hilbert^{(0)}\oplus \Hilbert^{(1)} = \CCC \oplus L^2(\RRR^3,\CCC^4)\,.
\ee
Correspondingly, the configuration space $\Q$ also consists of two sectors,
\be 
\Q = \Q^{(0)} \cup \Q^{(1)} = \{\emptyset\} \cup (\RRR^3\setminus \{\vzero\}) \,.
\ee
The process $(Q_t)_{t\in\RRR}$ that we construct moves in $\Q$. In the upper sector, it moves along a Bohmian trajectory until it hits the origin, at which time it jumps to the empty configuration $\emptyset$, where it remains for a random time and then jumps back to the upper sector, where it follows a Bohmian trajectory starting from $\vzero$ until it reaches $\vzero$ again, and so on. In particular, the process is piecewise deterministic (because the Bohmian trajectories are deterministic), and the only stochastic elements are the jumps between $\emptyset\in\Q^{(0)}$ and $\vzero\in\partial\Q^{(1)}$. More precisely, while the absorption events (jumps to $\emptyset$) are deterministic and occur whenever $Q_t$ reaches $\vzero\in\partial\Q^{(1)}$, the emission events (jumps from $\emptyset$) are stochastic in two ways: (i)~when they occur and (ii)~onto which trajectory the process jumps (because there can be several trajectories starting from $\vzero$ at the same time). 

The trajectories here are the solutions of Bohm's equation of motion for the Dirac equation \cite{Bohm53},
\be \label{guidingeq} 
\frac{\D\vQ(t)}{\D t} = \frac{\vj}{\rho}(\vQ(t))
\ee
(boldface symbols denoting 3d vectors) with probability current
\be 
\vj = (\psi^{(1)})^\dagger \valpha \psi^{(1)}\,, 
\ee
where $\psi^{(1)}$ is the $\Hilbert^{(1)}$-component of a wave function $\psi = (\psi^{(0)}, \psi^{(1)})$ in $\Hilbert$ and $\valpha = (\alpha_1, \alpha_2, \alpha_3)$ denotes the vector of the standard Dirac $\alpha$-matrices (see \eqref{eq:alpha}), and density 
\be 
\rho = j^0 = (\psi^{(1)})^\dagger \psi^{(1)} = \big|\psi^{(1)}\big|^2\,.
\ee
As mentioned, the process jumps to $\emptyset$ when it reaches $\vzero$. The other law needed to define the process (see Section~\ref{sec:sigma}) specifies the jump rate that applies whenever $Q_t=\emptyset$.
The process is designed so that
\be\label{equivariance}
Q_t \text{ is } |\psi_t|^2\text{-distributed}
\ee
at every time $t$. We will see in Section~\ref{sec:sigma} that the jump rate is in fact uniquely determined by the wish that \eqref{equivariance} holds for all $t$.

Away from the origin in $\RRR^3$, $H$ acts like the Dirac operator with a Coulomb potential of strength $q$,
\be\label{Coulomb}
(H\psi)^{(1)}(\vx) = \Bigl(-\I c\hbar\valpha \cdot \nabla + mc^2\beta + c\hbar\frac{q}{|\vx|} \Bigr) \psi^{(1)}(\vx)~~~~(\vx \neq \vzero)\,,
\ee
where $\beta = \mathrm{diag}(1,1,-1,-1)$ denotes the standard Dirac $\beta$-matrix.
On the other hand, $H$ couples the two sectors of $\Hilbert$, i.e., none of them stays invariant under the evolution generated by $H$. 
We assume that
\be\label{qrange}
\sqrt{3}/2 < |q| < 1\,.
\ee
For $|q|\leq \sqrt{3}/2$, there is no self-adjoint operator that couples the sectors and obeys \eqref{Coulomb}, and the case $|q|\geq 1$ was not studied in \cite{HT20}. We will give a full description of $H$, and write down the IBC, in Section~\ref{sec:H}. IBCs for Dirac operators on codimension-1 boundaries (as opposed to codimension 3 considered here) were studied in \cite{LN18,IBCdiracCo1}.

The construction of a Bell-type jump process for a similar model in curved space-time was outlined by one of us in \cite{timelike}. While that construction is very analogous in spirit to the one presented here, a relevant difference is that for the present model, a rigorously defined self-adjoint Hamiltonian $H$ is known, which allows for a precise and detailed description of the process that was not possible in \cite{timelike}.

It is of interest to compare our model with a non-relativistic variant \cite{ibc2a}, in which \eqref{Hilbertdef} is replaced by $\CCC\oplus L^2(\RRR^3,\CCC)$, $H$ with another operator $H_\nr$ (where the subscript nr stands for ``non-relativistic''), and \eqref{Coulomb} by
\be
(H_\nr\psi)^{(1)} = -\tfrac{\hbar^2}{2m} \Delta\,.
\ee
The natural variant of Bell's jump process for $H_\nr$ is described in \cite{bohmibc}. For $\psi$ from the domain of $H_\nr$, the probability current
\be 
\vj(\vx) = \tfrac{\hbar}{m} \Im \, \psi^{(1)}(\vx)^* \nabla \psi^{(1)}(\vx)
\ee
is, for every unit vector $\vomega\in\RRR^3$, of the form
\be 
\vj(r\vomega) = j_r \, \vomega + o(r^0)
\ee
as $r\searrow 0$ (i.e., to 0 from the right) with a constant $j_r$ independent of $r$ and $\vomega$. Put differently, the angular components of $\vj$ (perpendicular to $\vomega$) converge to 0 as $r \searrow 0$, while the radial component (along $\vomega$) converges to a generally nonzero value. As a consequence, the Bohmian trajectory, when drawn in spherical coordinates as in Figure~\ref{fig:nr}, hits the $r=0$ surface perpendicularly at nonzero speed.

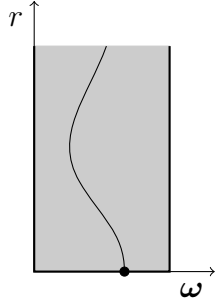
\begin{figure}[h]
\begin{center}
\begin{tikzpicture}[scale=1.2]
\draw (0,3) node[anchor=north east] {$r$};
\draw (2,0) node[anchor=north east] {$\vomega$};
\filldraw[fill=gray!40!white, draw=gray!40!white] (0,0) rectangle (1.5,2.5);
\draw[thick] (0,2.5) -- (0,0) -- (1.5,0) -- (1.5,2.5);
\draw[->] (0,0) -- (0,3);
\draw[->] (0,0) -- (2,0);
\filldraw (1,0) circle (0.5mm);
\draw (1,0) to[out=90,in=300] (0.5,1) to[out=120,in=250] (0.8,2.5);
\end{tikzpicture}
\end{center}
\caption{For the non-relativistic case, the trajectory in $\RRR^3$ before absorption is shown, represented in spherical coordinates, with only one of the two angles of $\vomega=(\varphi,\vartheta)$ drawn (shaded region = admissible values $r>0$, $0\leq \varphi<2\pi$, $0\leq \vartheta \leq \pi$). The trajectory ends at $r=0$ at a particular value of $\vomega$; the corresponding point $(0,\vomega)$ in the diagram is marked.}
\label{fig:nr}
\end{figure}

Certain features are different in the relativistic case of our $H$. Let 
\be \label{B}
B:=\sqrt{1-q^2}\,,
\ee
where $q$ is the strength of the Coulomb potential as in \eqref{Coulomb}; note that, due to \eqref{qrange}, $0<B<\frac12$.
We will argue that for $\psi$ from a certain subspace of $\Hilbert$,
a Bohmian trajectory $t\mapsto \vQ(t)\in\RRR^3$ that reaches $r=0$ does so at radial velocity $0$ and only after orbiting the $z$ axis infinitely many times. ($H$ is not rotationally invariant; it commutes with the $z$ component $J_z$ of angular momentum but not with other components.) In fact, as depicted in Figure~\ref{fig:r(t)}, almost surely,
\be\label{r(t)}
|\vQ(t)| \sim (\mathrm{const.}) \, |t-t_0|^{1/(1-2B)}
\ee
as $t\nearrow t_0$, where $t_0$ is the time it reaches $r=0$ and $\sim$ means asymptotically equal, i.e.
\be 
f(t) \sim g(t) ~~~:\Leftrightarrow~~~ \frac{f(t)}{g(t)}\to 1
 ~~~\Leftrightarrow~~~ f(t)=g(t)+o(g(t)).
\ee
Since $1/(1-2B)>1$, one would expect (and it is the case) that the curve, as a function of $t$, touches $r=0$ at $t_0$ with
\be 
\frac{\D r}{\D t}(t_0)=0\,.
\ee

\begin{figure}[h]
\begin{center}
\begin{tikzpicture}
\draw[->] (-0.4,0) -- (4,0);
\node at (3.8,-0.3) {$r$};
\draw[->] (0,-0.4) -- (0,2.6);
\node at (0.3,2.4) {$t$};
\draw (-0.2,2) -- (0,2);
\node at (-0.4,2) {$t_0$};
\draw[ultra thick] plot[domain=0:1] ({3*\x^1.75},{2-1.5*\x});
\end{tikzpicture}
\end{center}
\caption{Asymptotic dependence $r(t)$ of a Bohmian trajectory before absorption, drawn here for $q=\sqrt{187/196}$, i.e., $1/(1-2B)=7/4$}
\label{fig:r(t)}
\end{figure}
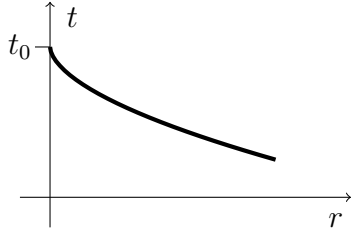

Moreover, the polar angle becomes constant at leading order in the limit $r \searrow 0$,
\be \label{theta(r)}
\vartheta(r) \sim (\mathrm{const.})\,,
\ee
while the dependence $\varphi(r)$ of the azimuthal angle on the radius is asymptotically of the form
\be\label{phi(r)}
\varphi(r) \sim (\mathrm{const.}) \, r^{-2B}
\ee
as $r\searrow 0$, see Figure~\ref{fig:phi(r)}.

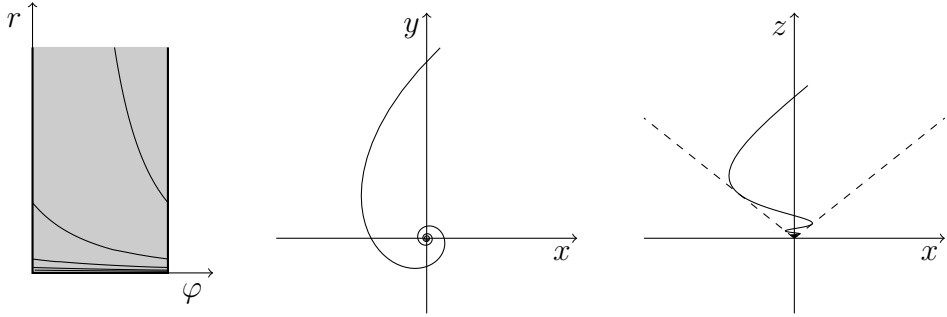
\begin{figure}[h]
\begin{center}
\begin{tikzpicture}[scale=1.2]
\draw (0,3) node[anchor=north east] {$r$};
\draw (2,0) node[anchor=north east] {$\varphi$};
\filldraw[fill=gray!40!white, draw=gray!40!white] (0,0) rectangle (1.5,2.5);
\draw[thick] (0,2.5) -- (0,0) -- (1.5,0) -- (1.5,2.5);
\draw[->] (0,0) -- (0,3);
\draw[->] (0,0) -- (2,0);
\draw plot[domain=0.39:1.25] ({1/\x^0.43},{2*\x});
\draw plot[domain=0.077:0.39] ({1/\x^0.43-1.5},{2*\x});
\draw plot[domain=0.0302:0.077] ({1/\x^0.43-3},{2*\x});
\draw plot[domain=0.0155:0.0302] ({1/\x^0.43-4.5},{2*\x});
\draw plot[domain=0.00914:0.01561] ({1/\x^0.43-6},{2*\x});
\end{tikzpicture}
~~~~
\begin{tikzpicture}
\draw[->] (-2,0) -- (2,0);
\draw[->] (0,-1) -- (0,3);
\node at (1.8,-0.2) {$x$};
\node at (-0.2,2.8) {$y$};
\draw plot[domain=1.5:25, samples=200] ({5*cos(\x r)/\x^1.67},{5*sin(\x r)/\x^1.67});
\end{tikzpicture}
~~~~
\begin{tikzpicture}
\draw[->] (-2,0) -- (2,0);
\draw[->] (0,-1) -- (0,3);
\node at (1.8,-0.2) {$x$};
\node at (-0.2,2.8) {$z$};
\draw[dashed] (0,0) -- (2,1.6);
\draw[dashed] (0,0) -- (-2,1.6);
\draw plot[domain=1.5:25, samples=200] ({5*cos(\x r)/\x^1.67},{4/\x^1.67});
\end{tikzpicture}
\end{center}
\caption{For the relativistic case, an asymptotic trajectory before absorption is shown for the same $q$ value as in Figure~\ref{fig:r(t)}. LEFT: Drawn in spherical coordinates; of the two angles of $\vomega=(\vartheta,\varphi)$, only the azimuthal angle $\varphi$ is shown. Its dependence on $r$ is given by \eqref{phi(r)}. MIDDLE: Drawn in Cartesian coordinates, seen along the $z$ axis. RIGHT: Drawn in Cartesian coordinates, seen along the $y$ axis (dashed = outline of the cone containing the curve).}
\label{fig:phi(r)}
\end{figure}

As a consequence of \eqref{theta(r)}, the asymptotic trajectory lies on a cone with (random) opening angle  $\vartheta(t_0)$, and $\varphi$ increases by an infinite amount before $r=0$ is reached, so it circles the $z$ axis infinitely often; see Figure~\ref{fig:perspective}. In particular, the trajectory does not have a limiting point on the 2-sphere $\{r=0\}$. Moreover, for each Hamiltonian $H$ from our family (i.e., for each choice of the parameters described in Section~\ref{sec:H}), there is a fixed sense of circling the $z$ axis: either, for all $\psi$, 
all trajectories asymptotically circle clockwise, or, for all $\psi$, all trajectories circle counter-clockwise. Likewise, the ``speed'' of orbiting, meaning here the exponent of $r^{-2B}$, is fixed by the choice of $H$ and does not depend on $\psi$. The time dependence $\varphi(t)$ can be obtained by inserting \eqref{r(t)} in \eqref{phi(r)}, which yields that
\be\label{phi(t)}
\varphi(t) \sim (\mathrm{const.}) \, |t-t_0|^{-2B/(1-2B)}
\ee
as $t\nearrow t_0$; see Figure~\ref{fig:phi(t)}.

\begin{figure}[h]
\begin{center}
\includegraphics[width=0.5\textwidth]{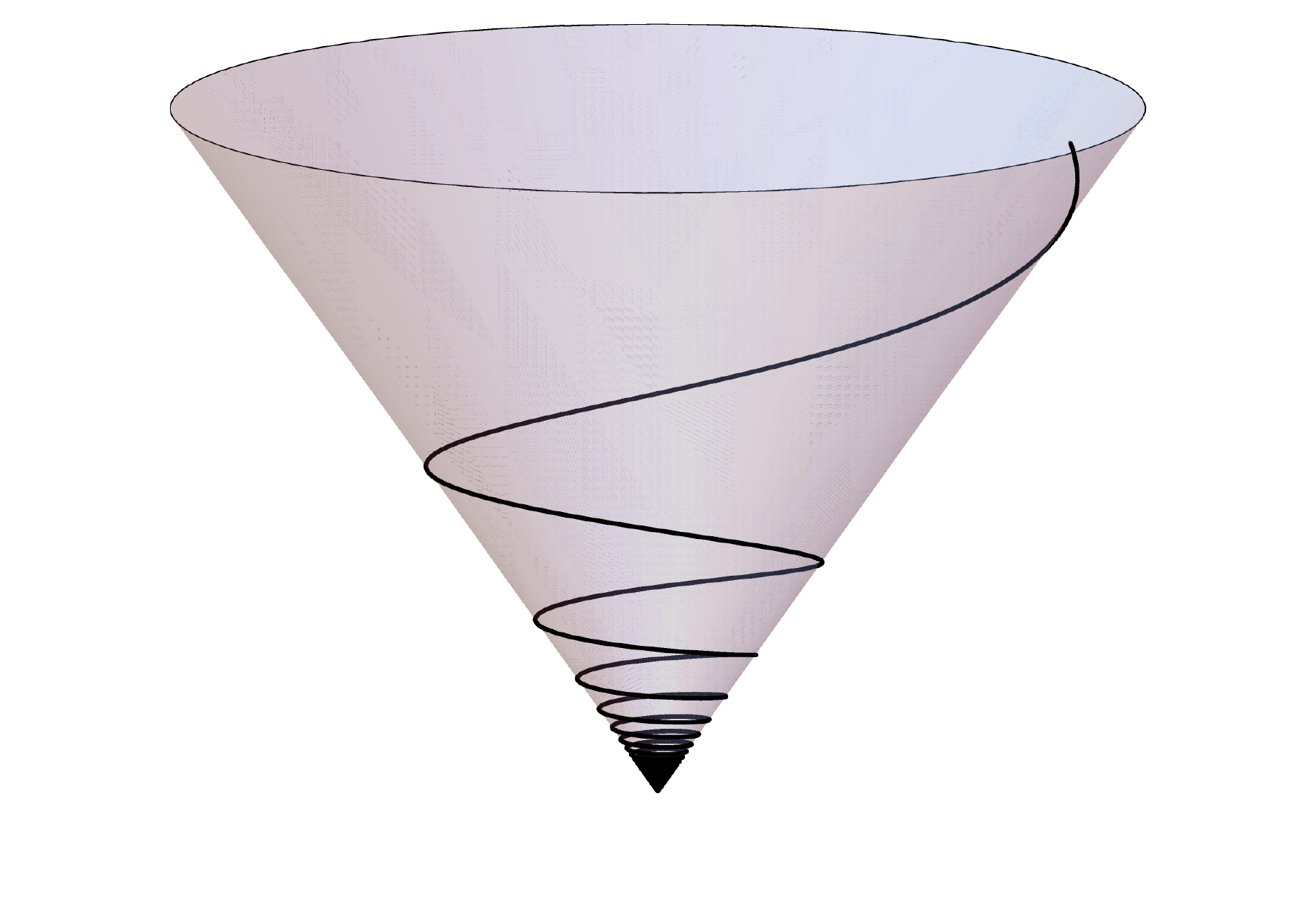}
\end{center}
\caption{The same curve as in Figure~\ref{fig:phi(r)} is shown as a curve in $\RRR^3$, seen in a perspective view. The curve lies on a cone of constant $\vartheta$ (not related to the light cone).}
\label{fig:perspective}
\end{figure}

The reverse trajectories that emerge from $r=0$ display the same behavior, i.e., \eqref{phi(r)} (with the reverse orientiation of the trajectory) and \eqref{r(t)} as $t\searrow t_0$. (If the ingoing trajectories circle clockwise, then so do the outgoing ones.)

\begin{table}
\begin{center}
\begin{tabular}{|r|c|c|}
\hline
&non-rel.&~~~rel.~~~\\\hline\hline
&&\\
$\displaystyle\frac{\D r}{\D t}(t_0)$ & $\neq 0$ & 0 \\
&&\\\hline
&&\\
$\vartheta(t_0)$ & const. & const. \\
&&\\\hline
&&\\
$\varphi(t_0)$ & const. & $\to\pm\infty$ \\
&&\\\hline
\end{tabular}
\end{center}
\caption{Comparison between the processes in the non-relativistic and the relativistic case; $t_0$ is the time of absorption or emission.}
\end{table}

This behavior, in particular the absence of a limit point on $\{r=0\}$, creates the following difficulty for the definition of a Bell-type jump process for this Hamiltonian. In the non-relativistic case, we could define a rate for jumping to the point $(0,\vomega)$, and then there is either a unique trajectory starting from there or a unique trajectory ending there. The rate was set to 0 when a trajectory ends there. Now, in the relativistic case, the trajectories emerging from $r=0$ do not possess a starting (limiting) point. We will be able to define a Bell-type jump process nevertheless by defining the rate for jumping onto a particular trajectory. In fact, the different trajectories can be characterized by their limiting $\vartheta(r=0)\equiv \vartheta_0$ values and their offsets (differences) $\varphi_0$ in the azimuthal angle. It turns out that the jump rate will be uniform over $\varphi_0$, so all trajectories with a given $\vartheta_0$ starting from $r=0$ at a given time are equally probable.

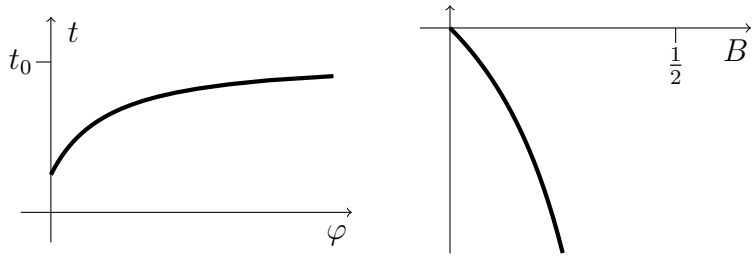
\begin{figure}[h]
\begin{center}
\begin{tikzpicture}
\draw[->] (-0.4,0) -- (4,0);
\node at (3.8,-0.3) {$\varphi$};
\draw[->] (0,-0.4) -- (0,2.6);
\node at (0.3,2.4) {$t$};
\draw (-0.2,2) -- (0,2);
\node at (-0.4,2) {$t_0$};
\draw[ultra thick] plot[domain=0.125:1] ({1/\x^0.75-1},{2-1.5*\x});
\end{tikzpicture}
~~~~
\begin{tikzpicture}
\draw[->] (-0.4,0) -- (4,0);
\draw (3,-0.2) -- (3,0);
\node at (3,-0.5) {$\frac{1}{2}$};
\node at (3.8,-0.3) {$B$};
\draw[->] (0,-3) -- (0,0.3);
\draw[ultra thick] plot[domain=0:0.25] ({6*\x},{-6*\x/(1-2*\x)});
\end{tikzpicture}
\end{center}
\caption{LEFT: Asymptotic dependence $\varphi(t)$ as in \eqref{phi(t)}, drawn for the same value of $B$ as in Figures~\ref{fig:r(t)}--\ref{fig:perspective}. RIGHT: The exponent in \eqref{phi(t)}, $-2B/(1-2B)$, as a function of $B$.}
\label{fig:phi(t)}
\end{figure}

We will only consider wave functions $\psi$ from a certain subspace $\widehat{D}\subset \Hilbert$ that is invariant under the time evolution; $\widehat{D}$ is the part of the domain $D$ of $H$ for which the component $\psi^{(1)}$ of $\psi=(\psi^{(0)},\psi^{(1)})$ in the upper sector lies in a certain angular momentum eigenspace (see Section~\ref{sec:H} for details). In fact, as we will see, the coupling between $\Hilbert^{(0)}$ and $\Hilbert^{(1)}$ happens only within $\widehat{D}$, so $\widehat{D}$ is the most relevant or interesting part of $D$. By focusing on $\widehat{D}$, we avoid unnecessarily tedious computation for extracting the qualitative behavior, which we believe will not change much for $\psi\in D \setminus \widehat{D}$.

The remainder of this paper is organized as follows. In Section~\ref{sec:H}, we report the relevant properties of $H$. In Section~\ref{sec:j}, we derive the asymptotic behavior of the current for $\psi\in\widehat{D}$. In Section~\ref{sec:Q}, we derive from that the (approximate) asymptotics of the Bohmian trajectories and justify the statements made above. In Section~\ref{sec:sigma}, we define the Bell-type jump process and justify the claim that it is equivariant.
In Section~\ref{sec:conclusion}, we conclude.

\section{The Hamiltonian}
\label{sec:H}

Let $\SSS^2$ denote the unit sphere in $\RRR^3$.
We will make use of a widely used orthonormal basis of $L^2(\SSS^2,\CCC^4)$, traditionally denoted
$\Phi^{\pm}_{m_j,\kappa_j}$, for which we have 
\begin{equation}
	L^2(\SSS^2,\CCC^4,d\Omega) = \bigoplus_{j \in \mathbb{N}_0+\frac{1}{2}} \bigoplus_{m_j = -j}^{j} \bigoplus_{\kappa_j= \pm(j+\frac{1}{2})} \Kilbert_{m_j\kappa_j} 
\end{equation}
with
\be\label{Kilbertdef}
\Kilbert_{m_j\kappa_j}=\mathrm{span}(\Phi^+_{m_j\kappa_j}, \Phi^-_{m_j\kappa_j})\,.
\ee
 The $\Phi^{\pm}_{m_j,\kappa_j}$ are simultaneous eigenvectors of $\vJ^2,K,J_3$ with $\vJ=\vL+\vS$ the total angular momentum and $K=\beta(2\vS\cdot \vL+1)$ the ``spin-orbit operator.''
In the standard representation of Dirac spin space, they are explicitly given by
\cite[(4.111)]{Tha91} 
\be\label{Phidef}
\Phi^+_{m_j,\mp (j+1/2)} = \begin{pmatrix} \I\Psi^{m_j}_{j\mp 1/2}\\0 \end{pmatrix}\,, ~~~~\Phi^-_{m_j,\mp(j+1/2)} = \begin{pmatrix}0\\\Psi^{m_j}_{j\pm1/2} \end{pmatrix}
\ee
with
\begin{subequations}\label{Psidef}
\begin{align}
\Psi^{m_j}_{j-1/2} &= \frac{1}{\sqrt{2j}} \begin{pmatrix} \sqrt{j+m_j} \: Y^{m_j-1/2}_{j-1/2}\\ \sqrt{j-m_j} \: Y^{m_j+1/2}_{j-1/2} \end{pmatrix} \\
\Psi^{m_j}_{j+1/2} &= \frac{1}{\sqrt{2j+2}} \begin{pmatrix}\sqrt{j+1-m_j} \: Y^{m_j-1/2}_{j+1/2} \\ -\sqrt{j+1+m_j} \: Y^{m_j+1/2}_{j+1/2}\end{pmatrix}
\end{align}
\end{subequations}
and $Y^m_l$ the usual spherical harmonics (defined for $l \in \mathbb{N}_0$ and $m \in \set{-l, \ldots, l}$), given by 
\be 
Y_{l}^m(\vartheta, \varphi) =  \sqrt{\frac{2l+1}{4\pi}} \sqrt{\frac{(l-m)!}{(l+m)!}} \: P_l^m(\cos(\vartheta)) \, \E^{\I m\varphi}, 
\ee
where 
\be 
P_l^m(x) = \frac{(-1)^m}{2^l l!}(1-x^2)^{m/2} \frac{\D^{l+m}}{\D x^{l+m}} (x^2-1)^l
\ee
are the associated Legendre polynomials. 

The Hamiltonian $H$ depends on parameters $g\in\CCC\setminus\{0\}$, $a_1,a_2,a_3,a_4\in\RRR$ with
\be\label{eq:a1toa4}
a_1 a_4-a_2 a_3 = 4B(1+q)\,,
\ee
and a fixed
\be\label{mjkappaj}
(\tilde m_j,\tilde\kappa_j) \in \sA:=\Bigl\{ (-\tfrac12,-1),~(-\tfrac12,1),~(\tfrac12,-1),~(\tfrac12,1) \Bigr\}\,.
\ee

As established in \cite{HT20} (using in particular results of \cite{Hog12,Gal17,GM19} about Dirac operators with Coulomb potential), the Hamiltonian $H$ and its domain $D$ have the following properties (which characterize the pair $(H,D)$ uniquely):
\begin{itemize}
\item For every $\psi\in D$, the upper sector is of the form
\be\label{short}
\psi^{(1)}(\vx) = c_{-}\, f^{-}_{\tilde m_j \tilde\kappa_j}\bigl(\tfrac{\vx}{|\vx|}\bigr)\, |\vx|^{-1-B} + \hspace{-3mm}\sum_{(m_j,\kappa_j)\in\sA} \hspace{-3mm} c_{+ m_j \kappa_j}\, f^+_{m_j \kappa_j}\bigl(\tfrac{\vx}{|\vx|}\bigr) \, |\vx|^{-1+B} + o(|\vx|^{-1/2})
\ee
as $\vx\to \vzero$ with (uniquely defined) short distance coefficients $c_{-},c_{+ m_j \kappa_j}\in \CCC$ and particular functions $f^{\pm}_{m_j \kappa_j}:\SSS^2\to\CCC^4$ given by
\begin{subequations}\label{fdef}
\begin{align}
f^+_{m_j \kappa_j} &= (1+q-B) \Phi^+_{m_j\kappa_j}-(1+q+B)\Phi^-_{m_j\kappa_j}\label{f+def}\\
f^{-}_{m_j \kappa_j} &= (1+q+B)\Phi^+_{m_j\kappa_j}- (1+q-B) \Phi^-_{m_j\kappa_j}\,.\label{f-def}
\end{align}
\end{subequations}
\item Every $\psi\in D$ obeys the IBC
\be\label{IBCgen}
a_1 \, c_{-} + a_2 \, c_{+ \tilde m_j \tilde\kappa_j} = g \, \psi^{(0)} \,,
\ee
and $H$ acts on $\psi\in D$ according to \eqref{Coulomb} and
\be\label{Hact0gen}
(H\psi)^{(0)}= g^* \,(a_3\, c_{-} + a_4 \, c_{+ \tilde m_j \tilde\kappa_j})\,.
\ee
\end{itemize}
We note that by rotational invariance of the Dirac operator with Coulomb potential, $H$ is block diagonal relative to the sum decomposition
\be\label{Hsummands}
\Hilbert \cong \widehat\Hilbert \oplus \bigoplus_{(j,m_j,\kappa_j) \neq (\tilde\jmath,\tilde m_j,\tilde \kappa_j)} L^2((0,\infty))\otimes \Kilbert_{m_j \kappa_j}
\ee
(recall \eqref{Kilbertdef} and note that $j$ is determined by $\kappa_j$ through $j=|\kappa_j|-\tfrac12$),
but, by means of the coupling in \eqref{IBCgen} and \eqref{Hact0gen}, not relative to
\be
\widehat\Hilbert=\Hilbert^{(0)}\oplus L^2((0,\infty)) \otimes \Kilbert_{\tilde m_j \tilde\kappa_j}\,.
\ee
Therefore, the subspace
\be
\widehat{D}:= D \cap \widehat{\Hilbert} 
\ee
is invariant under the time evolution generated by $H$. Henceforth, we will only consider $\psi$'s from this set. Since the coupling between $\Hilbert^{(0)}$ and $\Hilbert^{(1)}$ essentially happens within $\widehat{D}$ (it is independent of $c_{+ m_j \kappa_j}$ for $( m_j, \kappa_j)\neq (\tilde m_j, \tilde\kappa_j)$), we expect that the trajectories for other $\psi$'s will be qualitatively similar; although the formulas \eqref{r(t)}, \eqref{theta(r)}, \eqref{phi(r)}, \eqref{phi(t)} may not apply literally, slight modifications of them should.

For $\psi\in\widehat{D}$, we can simplify and refine \eqref{short} as follows:
\begin{multline}\label{shortDhat}
\psi^{(1)}(\vx) = c_{-}\, f^{-}_{\tilde m_j \tilde\kappa_j}\bigl(\tfrac{\vx}{|\vx|}\bigr)\, |\vx|^{-1-B} + c_{+ \tilde m_j \tilde\kappa_j}\, f^+_{\tilde m_j \tilde\kappa_j}\bigl(\tfrac{\vx}{|\vx|}\bigr) \, |\vx|^{-1+B} \\[2mm]
+ f^{-}_{\tilde m_j \tilde\kappa_j}\bigl(\tfrac{\vx}{|\vx|}\bigr)\, o(|\vx|^{-1/2}) 
+ f^{+}_{\tilde m_j \tilde\kappa_j}\bigl(\tfrac{\vx}{|\vx|}\bigr) \, o(|\vx|^{-1/2})\,.
\end{multline}
That is, apart from the fact that $c_{+m_j \kappa_j}=0$ for $(m_j,\kappa_j)\neq(\tilde m_j,\tilde\kappa_j)$, also the error terms must lie in $\Kilbert_{\tilde m_j \tilde\kappa_j}$.
Indeed, this follows from \eqref{short} by projecting to $\widehat{\Hilbert}$. In the following sections, we use \eqref{shortDhat} instead of \eqref{short}.

\section{The Current}
\label{sec:j}

Our goal in this section and the next is to compute the asymptotic behavior of the solutions of the equation of motion \eqref{guidingeq} in $\Q^{(1)}=\RRR^3\setminus \{\vzero\}$ that either reach $\vzero$ or come out of $\vzero$ at some time $t_0$. That is, we consider $t$ near $t_0$ and $r$ near 0. To this end, we replace $\psi_t$ by $\psi_{t_0}$ and determine the asymptotics of the solutions $\vQ(t)$ of \eqref{guidingeq} for fixed $\psi^{(1)} = (\psi_{t_0})^{(1)}$. We first need to establish the asymptotic behavior of the probability current
\be \label{j1}
\vj(\vx) = \psi^{(1)}(\vx)^\dagger \valpha \psi^{(1)}(\vx)
\ee
from the short-distance asymptotics of $\psi^{(1)}$ given in \eqref{shortDhat}. 
We already noted in the previous section that the coupling between the $0$--particle sector and the $1$--particle sector described by \eqref{IBCgen} and \eqref{Hact0gen} is independent of $c_{+ m_j \kappa_j}$ for $( m_j, \kappa_j)\neq (\tilde m_j, \tilde\kappa_j)$. 
Since we assume $\psi \in \widehat{D}$, we henceforth write $c_+$ instead of $c_{+\tilde m_j \tilde\kappa_j}$ for ease of notation.

\begin{prop}\label{prop:current}
For $\psi\in \widehat{D}$, 
the components of the probability current in spherical coordinates obey the following asymptotics as $\vx\to \vzero$:
\begin{subequations}\label{jasymptotics}
\begin{align} 
j_r(r\vomega) &= C_r  \: r^{-2} + o(r^{-3/2-B})\\[2mm]
j_\vartheta(r\vomega) &=  o(r^{-3/2-B})\\[2mm]
\label{jphiasympt} j_\varphi(r\vomega) &= {C_\varphi\, \sin \vartheta \: r^{-2-2B}   + \,\sin\vartheta \: O(r^{-2})}
\end{align}
\end{subequations}
where 
$\vomega\in \SSS^2$, 
{$C_r$ and $C_\varphi$ are real constants (that depend on $\psi$ but not on $r\vomega$),} 
$j_k := \ve_k \cdot \vj$ ($k=r,\vartheta,\varphi$), and $\ve_k$ is the unit vector in the $k$ direction, 
\begin{subequations}
\begin{align}
\ve_r&= (\sin \vartheta \cos \varphi, \sin \vartheta \sin \varphi, \cos \vartheta)= \frac{\vx}{|\vx|}=\vomega\\
\ve_\vartheta &= (\cos \vartheta \cos \varphi, \cos \vartheta \sin \varphi, -\sin \vartheta)\\[3mm]
\ve_\varphi &= (-\sin \varphi, \cos \varphi, 0)\,.
\end{align}
\end{subequations}
More explicitly, we have that
\begin{subequations}\label{Jasymptotics}
	\begin{align} 
	j_r(r\vomega) &= \frac{2(1+q) B}{\pi} \, \Im[c^*_- c_+]   \: r^{-2} + o(r^{-3/2-B})\\[2mm]
	j_\vartheta(r\vomega) &=  o(r^{-3/2-B})\\[2mm]
	\nonumber j_\varphi(r\vomega) &= - \frac{q(1+q)}{\pi} \, |c_-|^2 \, \mathrm{sgn}(\tilde m_j \tilde\kappa_j) \sin \vartheta \: r^{-2-2B}  \\
	\label{Jphiasympt} & \hspace{1cm} - \frac{2(1+q)}{\pi} \, \Re[c^*_- c_+] \, \mathrm{sgn}(\tilde m_j \tilde\kappa_j) \sin \vartheta \:  r^{-2} \\
	\nonumber & \hspace{1cm} -   \frac{q(1+q)}{\pi} \, |c_+|^2 \, \mathrm{sgn}(\tilde m_j \tilde\kappa_j) \sin \vartheta \: r^{-2+2B}  \,  + {\sin\vartheta \: o(r^{-3/2-B})} \,.
	\end{align}
\end{subequations}
\end{prop}

\begin{proof}
From \eqref{shortDhat} and \eqref{j1}, using $c_{+ m_j \kappa_j} = 0$ for $( m_j, \kappa_j)\neq (\tilde m_j, \tilde\kappa_j)$,
\begin{subequations}
\begin{align}
\vj(r\vomega)
&=  \psi^{(1)}(r\vomega)^\dagger \, \valpha \, \psi^{(1)}(r\vomega)\\[2mm]
&=|c_{-}|^2 \braket{f_{\tilde{m}_j, \tilde{\kappa}_j}^-(\vomega),\valpha f_{\tilde{m}_j, \tilde{\kappa}_j}^-(\vomega)}_{\CCC^4} \: r^{-2-2B} \nonumber\\
&\hspace{1cm}+2\, \Re\Bigl[ c_{-}^*c_{+ } \braket{f_{\tilde{m}_j, \tilde{\kappa}_j}^-(\vomega),\valpha f^+_{m_j\kappa_j}(\vomega)}_{\CCC^4}\Bigr] \: r^{-2} \label{jasympt}\\
&\hspace{1cm}+ |c_{+}|^2 \braket{f_{\tilde{m}_j, \tilde{\kappa}_j}^+(\vomega),\valpha f_{\tilde{m}_j, \tilde{\kappa}_j}^+(\vomega)}_{\CCC^4} \: r^{-2+2B} \nonumber\\
&\hspace{1cm}+ {\sum_{\nu,\pi=\pm} \braket{f_{\tilde{m}_j, \tilde{\kappa}_j}^\nu(\vomega),\valpha f_{\tilde{m}_j, \tilde{\kappa}_j}^\pi(\vomega)}_{\CCC^4} \: o(r^{-3/2- B})} \nonumber\,.
\end{align}
\end{subequations}
In Lemma~\ref{lem:f-+omega} below, we evaluate the coefficients of $r^{-2-2B}$ and $r^{-2+2B}$ and in particular show that they vanish in the $r$ and $\vartheta$ components. Afterwards, in Lemma~\ref{lem:f-f+}, we evaluate the coefficient of $r^{-2}$ and in particular show that it is independent of $\vomega$ in the $r$ component and vanishes in the $\vartheta$ component. Lemmas~\ref{lem:f-+omega} and \ref{lem:f-f+} also show that all terms of the $\varphi$ component of \eqref{jasympt} contain a factor of $\sin\vartheta$. This yields \eqref{jasymptotics}. Inserting the precise results for the coefficients in Lemma \ref{lem:f-+omega} and Lemma \ref{lem:f-f+} we arrive at \eqref{Jasymptotics}.
We remark {about the last two lines of \eqref{jasympt} that} it depends on $B$ which of the exponents $-2+2B$ and $-3/2-B$ is greater; for $B>1/6$, $-2+2B$ is greater, so $r^{-2+2B}<r^{-3/2-B}$, and the $r^{-2+2B}$ term could be included in the $o(r^{-3/2-B})$.
\end{proof}

\begin{lem}\label{lem:f-+omega}
	For every $\vomega\in\SSS^2$, we have that
	\begin{subequations} \label{jasymptotics2}
	\begin{align}
	\braket{f_{\tilde{m}_j, \tilde{\kappa}_j}^\mp(\vomega),\alpha_r f_{\tilde{m}_j, \tilde{\kappa}_j}^\mp(\vomega)}_{\CCC^4} &=0, \label{jrasympt1}\\
	\braket{f_{\tilde{m}_j, \tilde{\kappa}_j}^\mp(\vomega),\alpha_\vartheta f_{\tilde{m}_j, \tilde{\kappa}_j}^\mp(\vomega)}_{\CCC^4} &=0, \label{jthetaasympt1}\\
	\braket{f_{\tilde{m}_j, \tilde{\kappa}_j}^\mp(\vomega),\alpha_\varphi f_{\tilde{m}_j, \tilde{\kappa}_j}^\mp(\vomega)}_{\CCC^4} &=-   \frac{q(1+q)}{\pi} \mathrm{sgn}({\tilde{m}_j \tilde{\kappa}_j}) \sin\vartheta \,, \label{jphiasympt1}
	\end{align}
	\end{subequations}
	where the $f_{\tilde{m}_j, \tilde{\kappa}_j}^\pm$ were defined in \eqref{fdef} and $\alpha_{k } := \ve_{k } \cdot \valpha$ for $k=r,\vartheta,\varphi$.
\end{lem}

\begin{proof} We omit the subscript $\tilde{m}_j, \tilde{\kappa}_j$ for ease of notation. By \eqref{f-def} (using that all components of $\valpha$ are self-adjoint),
\begin{align}
\braket{f^\mp(\vomega),\valpha f^\mp(\vomega)}_{\CCC^4}
&=(1+q+B)^2\braket{\Phi^\pm(\vomega),\valpha \Phi^\pm(\vomega)}_{\CCC^4} \nonumber\\
&~~~ -2 (1+q+B)(1+q-B)\, \Re\braket{\Phi^+(\vomega),\valpha\Phi^-(\vomega)}_{\CCC^4}  \label{expanded}\\ 
&~~~ +(1+q-B)^2\braket{\Phi^\mp(\vomega),\valpha \Phi^\mp(\vomega)}_{\CCC^4} \,. \nonumber
\end{align}
Since in the standard representation
\be \label{eq:alpha}
\valpha = \begin{pmatrix}0&\vsigma\\  \vsigma&0\end{pmatrix}
\ee
with $\vsigma=(\sigma_1,\sigma_2,\sigma_3)$ the Pauli matrices, we can read off from the form \eqref{Phidef} that
\be 
\braket{\Phi^{\pm}(\vomega),\valpha \Phi^{\pm}(\vomega)}_{\CCC^4}=\vzero
\ee
for every $\vomega\in\SSS^2$. Thus, the first and the third line of \eqref{expanded} vanish identically.

We will now compute
\be 
\braket{\Phi^+(\vomega),\alpha_k\Phi^-(\vomega)}_{\CCC^4}=-\I \braket{\Psi^{m_j}_{j\mp 1/2}(\vomega),\ve_k\cdot\vsigma \Psi^{m_j}_{j\pm1/2}(\vomega)}_{\CCC^2}\,.
\ee
%
%
For us $j = 1/2$, so we recall that the first few spherical harmonics are
\be 
Y^0_0(\vartheta, \varphi) = \frac{1}{\sqrt{4\pi}}\,, \ \  Y_1^0(\vartheta, \varphi) = \sqrt{\frac{3}{4\pi}} \cos \vartheta\,, \ \  Y_1^{\pm1}(\vartheta, \varphi) = \mp \sqrt{\frac{3}{8\pi}} \sin \vartheta \E^{\pm \I \varphi}
\ee
and verify:
\begin{subequations}
\begin{align}
\braket{\Psi^{\pm 1/2}_{1}(\vomega),\sigma_1 \Psi^{\pm 1/2}_{0}(\vomega)}_{\CCC^2}^* = \braket{\Psi^{\pm 1/2}_{0}(\vomega),\sigma_1 \Psi^{\pm 1/2}_{1}(\vomega)}_{\CCC^2}
&=\frac{1}{4\pi} \sin \vartheta \, \E^{\pm \I \varphi} \\
\braket{\Psi^{\pm 1/2}_{1}(\vomega),\sigma_2 \Psi^{\pm 1/2}_{0}(\vomega)}_{\CCC^2}^* = \braket{\Psi^{\pm 1/2}_{0}(\vomega),\sigma_2 \Psi^{\pm 1/2}_{1}(\vomega)}_{\CCC^2}
&=\mp \frac{\I}{4\pi} \sin \vartheta \, \E^{\pm \I \varphi} \\
\braket{\Psi^{\pm 1/2}_{1}(\vomega),\sigma_3 \Psi^{\pm 1/2}_{0}(\vomega)}_{\CCC^2}^* = \braket{\Psi^{\pm 1/2}_{0}(\vomega),\sigma_3 \Psi^{\pm 1/2}_{1}(\vomega)}_{\CCC^2}
&=\frac{1}{4\pi }  \cos \vartheta\,. 
\end{align}
\end{subequations}
Thus, we arrive at
\begin{subequations} \label{scp}
\begin{align}
\braket{\Phi_{\tilde{m}_j, \tilde{\kappa}_j}^+(\vomega),\alpha_r\Phi_{\tilde{m}_j, \tilde{\kappa}_j}^-(\vomega)}_{\CCC^4} &= - \frac{\I}{4\pi} \label{scpr}\\
\braket{\Phi_{\tilde{m}_j, \tilde{\kappa}_j}^+(\vomega),\alpha_\vartheta\Phi_{\tilde{m}_j, \tilde{\kappa}_j}^-(\vomega)}_{\CCC^4} &= 0 \label{scptheta}\\
\braket{\Phi_{\tilde{m}_j, \tilde{\kappa}_j}^+(\vomega),\alpha_\varphi\Phi_{\tilde{m}_j, \tilde{\kappa}_j}^-(\vomega)}_{\CCC^4} &= \mathrm{sgn}(\tilde{m}_j \tilde{\kappa}_j) \frac{1}{4 \pi} \sin \vartheta\,. \label{scpphi}
\end{align}
\end{subequations}
Now \eqref{jthetaasympt1} follows from \eqref{scptheta}, \eqref{jrasympt1} follows from the fact that \eqref{scpr} has vanishing real part, and \eqref{jphiasympt1} is obtained from the middle row of \eqref{expanded} and \eqref{scpphi}. 
\end{proof}

We can also read off from Lemma~\ref{lem:f-+omega} that the leading order coefficient of $j_\varphi$ in \eqref{Jphiasympt} is given by
\be \nonumber
 - |c_-|^2\, \mathrm{sgn}({\tilde{m}_j \tilde{\kappa}_j}) \frac{q(1+q)}{\pi} \sin\vartheta\,,
\ee
showing that the sign of $j_\varphi$ near $r=0$ is fixed for fixed parameters $q,\tilde m_j,\tilde\kappa_j$.

\begin{lem}\label{lem:f-f+}
	For every $\vomega\in\SSS^2$, we have that
	\begin{subequations} \label{jasymptotics3}
	\begin{align}
	\braket{f_{\tilde{m}_j, \tilde{\kappa}_j}^-(\vomega),\alpha_r f_{\tilde m_j,\tilde \kappa_j}^+(\vomega)}_{\CCC^4} &=- \I \frac{ (1+q) B}{\pi} \,, \label{jrasympt2}\\
	\braket{f_{\tilde{m}_j, \tilde{\kappa}_j}^-(\vomega),\alpha_\vartheta f_{\tilde m_j, \tilde \kappa_j}^+(\vomega)}_{\CCC^4} &= 0 \,, \label{jthetaasympt2}\\
	\braket{f_{\tilde{m}_j, \tilde{\kappa}_j}^-(\vomega),\alpha_\varphi f_{\tilde m_j, \tilde \kappa_j}^+(\vomega)}_{\CCC^4} &= -  \frac{1+q}{\pi} \, \mathrm{sgn}({\tilde{m}_j \tilde{\kappa}_j}) \sin\vartheta\,, \label{jphiasympt2}
	\end{align}
	\end{subequations}
	where the $f_{\tilde{m}_j, \tilde{\kappa}_j}^\pm$ were defined in \eqref{fdef} and $\alpha_{k } = \ve_{k } \cdot \valpha$ for $k=r,\vartheta,\varphi$.
\end{lem}

\begin{proof}
We omit the subscript $\tilde{m}_j, \tilde{\kappa}_j$ again and argue exactly as in the proof of Lemma~\ref{lem:f-+omega} to find that $\braket{f^-(\vomega),\valpha f^+(\vomega)}_{\CCC^4}$ equals
\begin{equation}
 - (1+q+B)^2\braket{\Phi^+(\vomega),\valpha \Phi^-(\vomega)}_{\CCC^4}  - (1+q-B)^2\braket{\Phi^+(\vomega),\valpha \Phi^-(\vomega)}_{\CCC^4}^*\,.
\end{equation}
Now, using \eqref{scp} and $B = \sqrt{1- q^2}$, the claim follows. 
\end{proof}

\section{The Trajectories}
\label{sec:Q}

From the asymptotic behavior \eqref{jasymptotics} resp.~\eqref{Jasymptotics} of the current and the fact that 
the probability density
\be \label{rho}
\rho(\vx)=\psi^{(1)}(\vx)^\dagger\, \psi^{(1)}(\vx)
\ee
is asymptotically proportional to $|\vx|^{-2-2B}$, we will now draw conclusions about the asymptotic Bohmian trajectories. 

To this end, we 
study approximate solutions of \eqref{guidingeq} by neglecting the time dependence of the velocity field $\vj/\rho$ on the right-hand side of \eqref{guidingeq}. This means, if $t \mapsto \psi_t = \E^{-\I H t} \psi_0$ denotes the (strongly differentiable) time--evolution of $\psi_0\in D$ governed by our Hamiltonian $H$, we make the simplifying assumption that $\vQ(t)$ is guided by a \emph{constant} velocity field; that is, we approximate $\psi_t \approx \psi_{t_0}$ and solve the differential equation
\be \label{guidingeqsimple}
 \frac{\D\vQ(t)}{\D t} = \left. \frac{\vj}{\rho}\right\vert_{t=t_0}(\vQ(t))
\ee
instead of \eqref{guidingeq} for times $t$ close to $t_0$. This approximation has already been employed in prior studies of Bohmian trajectories in the context of IBCs \cite{bohmibc}; see Remark \ref{rmk:rigorous} below for a possible general strategy of rigorously justifying it.

\begin{prop}\label{prop:traj}
Let $\psi_0\in \widehat{D}$ and $t_0 \in \RRR$ be any time for which 
\begin{equation} \label{special}
\Im[c^*_-(t_0) c_{+ }(t_0)] \neq 0
\end{equation}
(in particular, $c_-(t_0)\neq 0$).
By simple time shifts, we may assume without loss of generality that $t_0 = 0$ and drop the argument $t_0$ in \eqref{special} from now on.

Then the trajectories $t\mapsto \vQ(t)$ solving \eqref{guidingeqsimple} (as an approximation of \eqref{guidingeq}) and reaching $r=0$ at time $t_0 = 0$ (or emanating from $r=0$ at $t_0= 0$) can occur only if $\Im[c_-^*c_+]< 0$ (resp., $\Im[c_-^*c_+]> 0$) and obey for $t< 0$ (resp., $t> 0$) in spherical coordinates the asymptotics
\begin{subequations}\label{Qasymptotics}
\begin{align}
r(t)&= \  \left[2 B \, (1-2B) \, \frac{\big| \Im[c^*_- c_+] \big| }{|c_-|^2} \right]^{\frac{1}{1-2B}}\: |t|^{\frac{1}{1-2B}} 
+ O\Bigl(|t|^{\min\{\frac{1+2B}{1-2B}, \frac{3/2-B}{1-2B}\}}\Bigr) 
\label{eq:rasymp}\\[2mm]
\vartheta(t)&= \ \vartheta_0 + o \big(|t|^{\frac{1}{2}}\big)  \label{eq:thetaasymp}\\
\varphi(t)& = \ \varphi_0 - q \, \mathrm{sgn}(\tilde m_j \tilde \kappa_j) \, 4 B^2 \left[ \left[\frac{2B}{1-2B}\right]^{2B} \hspace{-1mm}  \frac{|c_-|^2}{\big| \Im[c^*_- c_+] \big| } \right]^{\frac{1}{1-2B}} \: |t|^{\frac{-2B}{1-2B}}  \label{eq:phiasymp}  \\[1mm] & \hspace{3cm}+ \ C_{H, c_\pm}\log |t| 
+ O\Bigl(|t|^{\min\{\frac{2B}{1-2B},\frac{1}{2}\}}\Bigr) 
\nonumber
\end{align}
\end{subequations} 
as $t\to 0$ with some (unique) constants $0\leq\varphi_0 <2\pi$ and $0\leq \vartheta_0\leq \pi$. 
Here, $C_{H, c_\pm}$ denotes a constant depending on the chosen Hamiltonian $H$ (i.e., on $q,\tilde{m}_j, \tilde{\kappa}_j$) and the short--distance coefficients $c_\pm$ of $\psi_0 \in \widehat{D}$. Moreover,
\be
\frac{\D r}{\D t} = O\Bigl(|t|^{\frac{2B}{1-2B}}  \Bigr) \stackrel{t\to 0}{\longrightarrow} 0\,.
\ee
\end{prop}

Recall that $\sqrt{3}/2 <|q| < 1$ and thus $B = \sqrt{1-q^2} \in (0,1/2)$ as defined in~\eqref{B}. It follows that for every $B$, the error term in \eqref{eq:rasymp} has exponent greater than $1/(1-2B)>0$ and thus is smaller than the explicitly given first term. Likewise in \eqref{eq:phiasymp}, the error term is actually smaller than the terms before because $2B/(1-2B)$ is always positive.

The condition \eqref{special} can be thought of as ensuring non-degeneracy of the Bohmian dynamics. Since $\widehat{D}$ is invariant under the time evolution generated by $H$, all the other short--distance coefficients apart from $c_\pm$ remain zero for all times. 

Moreover, observe that, by plugging \eqref{eq:rasymp} into \eqref{eq:thetaasymp} and \eqref{eq:phiasymp}, we arrive at \eqref{theta(r)} and \eqref{phi(r)}, respectively. Since the leading order coefficient of \eqref{eq:phiasymp} is given by 
$-q \,  \mathrm{sgn}(\tilde{m}_j \tilde{\kappa}_j)$ times a positive factor depending also on $c_{\pm}$, we see that the sense of circling
the $z$--axis depends on the choice of the Hamiltonian (viz., on $q, \tilde{m}_j,\tilde{\kappa}_j$) but not on $\psi$, while the speed of circulation (meaning not just the exponent of $|t|$ but also the prefactor) depends on $\psi$ but is the same for all trajectories.

\begin{rmk} {\rm (On the approximation by a constant velocity field)}\label{rmk:rigorous} \\
	The approximate form \eqref{guidingeqsimple} of the equation of motion \eqref{guidingeq} has already been used in the derivation of Bohmian trajectories for the non-relativistic case in \cite{bohmibc}. Although the simplified ODE \eqref{guidingeqsimple} (and its non-relativistic analog in \cite{bohmibc}) most likely yield the correct leading order behavior of Bohmian trajectories shortly after (before) particle creation (annihilation), both \cite{bohmibc} and the present work are lacking a rigorous justification of this approximation. In the following, we shall thus briefly outline a potential general strategy of how one could prove the validity of approximating the full guiding equation \eqref{guidingeq} by the one with a constant velocity field \eqref{guidingeqsimple}. We will focus on the present relativistic setting, but the principal argument can immediately be translated to the non-relativistic setting \cite{bohmibc}. 
	
	The basic idea to make the approximation rigorous is to show that for $\psi_0 \in \widehat{D}$, the three terms in the asymptotic expansion for the $1$--particle component of the time--evolved wave function
	\begin{equation}
\label{eq:psi t}
\psi^{(1)}_t(r \vomega) = c_{-}(t)\, f^{-}_{\tilde m_j \tilde\kappa_j}(\vomega)\, r^{-1-B} + c_{+}(t)\, f^+_{m_j \kappa_j}(\vomega) \, r^{-1+B} + o_t(r^{-1/2})
	\end{equation}
are well--behaved in $t$. More precisely, one needs to show that (i) $c_-(t)$ is a $C^1$--function of time, (ii) $c_+(t)$ is a $C^1$--function of time, and (iii) the implicit constant in $o_t(|\vx|^{-1/2})$ is uniformly bounded for small enough times. First, assuming that we have $a_1=1$, $a_4 = 4B(1+q)$, and $a_2 = a_3 = 0$ in \eqref{eq:a1toa4}, the IBC \eqref{IBCgen} yields that $c_- \in C^1$ since $\psi_t^{(0)}$ is $C^1$ in time and we have proven (i). Note that, if we had chosen different $a_1, ... , a_4$, we could have drawn the same conclusion for a certain linear combination of $c_-$ and $c_+$. For (ii) we propose to take the scalar product of $\psi^{(1)}_t$ with $g_t(r\vomega) =  f^+_{m_j \kappa_j}(\vomega)x(t)^{-1/2} \mathbf{1}_{\{ r < x(t) \}}$ with $x(t) \to 0$ as $t \to 0$. Using that $\big\vert   \langle (\psi^{(1)}_t - \psi^{(1)}_0), g_t \rangle \big\vert \le C |t|$ in combination with $c_-$ being $C^1$, one should be able to deduce the same regularity for $c_+$ by taking $x(t) \to 0$ as $t \to 0$ arbitrarily slow. For (iii) we note that the $o_t(r^{-1/2})$--error in \eqref{eq:psi t} originates from integrating a $H_0^1((0,\infty))$ function from $0$ to $r$ by the fundamental theorem of calculus \cite{GM19} and dividing by $r$ afterwards. Therefore, in order to show the error term to be bounded uniformly in short times, one could employ Sobolev--to--Sobolev estimates showing that the time evolution $\E^{-\I Ht}$ is a bounded operator from one Sobolev space to another, uniformly for times $t$ in compact intervals (see, e.g.,~\cite{SobSob}). 
\end{rmk}

It remains to give the proof of Proposition \ref{prop:traj}. 
\begin{proof}[Proof of Proposition \ref{prop:traj}]
By the short--distance asymptotics {\eqref{shortDhat} or \eqref{short}}, we have that
\begin{multline} \label{eq:rhoexp}
\rho(r\vomega) = |c_-|^2 \, \bigl| f_{\tilde m_j \tilde\kappa_j}^-(\vomega) \bigr|^2 \, r^{-2-2B} \\
+ 2 \Re\big[c_-^*c_+\langle f^-_{\tilde{m}_j, \tilde{\kappa}_j}(\vomega), f^+_{\tilde{m}_j, \tilde{\kappa}_j}(\vomega)\rangle \big] r^{-2} + O\bigl( r^{\min\{-2+2B,-3/2-B\}} \bigr)\,.
\end{multline}
An easy computation yields that
\begin{equation}
\langle \Phi_{\tilde m_j \tilde\kappa_j}^\pm(\vomega), \Phi_{\tilde m_j \tilde\kappa_j}^\mp(\vomega)\rangle_{\CCC^4} = 0 \qquad \text{and} \qquad \big\vert \Phi_{\tilde m_j \tilde\kappa_j}^\pm(\vomega)  \big\vert^2 = \frac{1}{4 \pi}\,,
\end{equation} 
which in particular shows that the $r^{-2}$ term in \eqref{eq:rhoexp} is independent of $\vomega$, and allows us to infer that
\begin{equation}
\bigl| f_{\tilde m_j \tilde\kappa_j}^-(\vomega) \bigr|^2 = \frac{1+q}{\pi} \qquad \text{and} \qquad \langle f^-_{\tilde{m}_j, \tilde{\kappa}_j}(\vomega), f^+_{\tilde{m}_j, \tilde{\kappa}_j}(\vomega)\rangle = \frac{q(1+q)}{\pi}\,.
\end{equation}
In this way we arrive at
\begin{equation} \label{eq:rhoasympt}
\rho(r\vomega) = |c_-|^2 \frac{1+q}{\pi} r^{-2-2B} + 2 \Re[c_-^*c_+ ] \frac{q(1+q)}{\pi} r^{-2} + O\bigl( r^{\min\{-2+2B,-3/2-B\}} \bigr)\,,
\end{equation}
where the explicit terms are independent of $\vomega$. Combining the asymptote \eqref{eq:rhoasympt} with Proposition \ref{prop:current} (and using that
\be
\frac{1}{A+\varepsilon}= \frac{1}{A}-\frac{\varepsilon}{A^2}+o(\varepsilon)= \frac{1}{A} + O(\varepsilon)
\ee
as $\varepsilon\to 0$ for $A\neq 0$ independent of $\varepsilon$), we obtain from the simplified equation of motion \eqref{guidingeqsimple} the following asymptotic system of ODEs for the spherical coordinates $(r(t), \vartheta(t), \varphi(t))$ of $\vQ(t)$,
\begin{subequations}\label{eq:ODEs}
	\begin{align}
		\frac{\D r}{\D t}&= \  2B \frac{\Im [c_-^* c_+]}{|c_-|^2} r^{2B} 
		+ O\bigl(r^{\min\{4B,1/2+B\}}\bigr) 
		\label{eq:rode}\\[2mm]
r	\, \frac{	\D \vartheta}{\D t}&= \  o (r^{1/2+B})  \label{eq:thetaode}\\[2mm]
		r  \, \frac{\D \varphi}{\D t}& = \ - q \, \mathrm{sgn}(\tilde m_j \tilde \kappa_j)   + \tilde{C}_{H, c_\pm}  r^{2B}
		+  {O\bigl(r^{\min\{4B,1/2+B\}}\bigr)} \,. 
		\label{eq:phiode} 
	\end{align}
\end{subequations} 
As in \eqref{eq:phiasymp}, $\tilde{C}_{H, c_\pm}$ denotes a constant depending on the choice of Hamiltonian $H$ (i.e., on $q$, $\tilde{m}_j$, and $\tilde{\kappa}_j$) and the short--distance coefficients $c_\pm$. In the last equation, we have already divided by $\sin\vartheta$.

We are now left with the task of solving the system \eqref{eq:ODEs}. Using the initial condition $r(0)=0$, the first equation \eqref{eq:rode} can be integrated by separation of variables, leaving us with 
 \begin{equation}
 	r(t)^{1-2B}\left( 1
	+ O\big(r(t)^{\min\{2B,1/2-B\}}\big) 
	\right)= \  \left[2 B \, (1-2B) \, \frac{\big| \Im[c^*_- c_+] \big| }{|c_-|^2} \right] \: |t|
 \end{equation}
for $\mathrm{sgn}(t)=\mathrm{sgn}(\Im[c_-^*c_+])$. Generally, from a relation of the form $t= cr^\alpha+O(r^\beta)$ with $0<\alpha<\beta$ and $c,t,r>0$, we can conclude that every $O(r^\gamma)$ is an $O(t^{\gamma/\alpha})$ and vice versa for every $\gamma>0$. Thus, $r^\alpha=c^{-1}t+O(t^{\beta/\alpha})$ and $r=c^{-1/\alpha}t^{1/\alpha} + O(t^{1/\alpha-1+\beta/\alpha})$, which yields \eqref{eq:rasymp}. 
For \eqref{eq:thetaode}, we make the change of variables $t \to r(t)$, insert the differential \eqref{eq:rode} to obtain that $\D\vartheta/\D r= o(r^{-1/2-B})$, and again integrate by separation of variables, where we now use the initial condition $\vartheta(0) = \vartheta_0$. In this way, we arrive at \eqref{eq:thetaasymp} after inverting the change of variables with the aid of \eqref{eq:rasymp}. 
In order to get \eqref{eq:phiasymp} from \eqref{eq:phiode}, we pursue the same strategy, i.e., replace $t \to r(t)$ and integrate by separation of variables. However, this time we need to choose the initial condition according to $\varphi(r= r_0) = \tilde{\varphi}_0$ for some sufficiently small but fixed $r_0> 0$ and $\tilde{\varphi}_0 \in \RRR$. Absorbing $\tilde{\varphi}_0$ and all terms depending only on $r_0$ into a new constant $\varphi_0 \in \RRR$, we arrive at \eqref{eq:phiasymp}, again after inverting the change of variables with the aid of \eqref{eq:rasymp}. 
\end{proof}

\section{The Jump Process}
\label{sec:sigma}

\subsection{Definition}
\label{sec:defQ}

We define the process $(Q_t)$ for $t\geq \tau$ for some time $\tau$ regarded as the initial time. Given that, as we will argue in Section~\ref{sec:equivariance}, the process is equivariant (i.e., $|\psi_t|^2$ distributed at every $t$), it follows that the processes defined for $\tau_1$ and $\tau_2>\tau_1$ are equal in distribution on $[\tau_2,\infty)$, so (by the Kolmogorov extension theorem) the processes for all $\tau$'s can be combined into a single process $(Q_t)_{t\in\RRR}$ defined on the whole time axis.

Here is the definition of the process. We assume that the initial wave function $\psi_\tau$ lies in $\widehat{D}$; it follows that $\psi_t\in\widehat{D}$ for all $t$. The initial configuration $Q_\tau$ is chosen to be $|\psi_\tau|^2$ distributed. Once $Q_t\in\Q^{(1)}=\RRR^3\setminus \{\vzero\}$, it follows the Bohmian trajectory, i.e., the equation of motion \eqref{guidingeq}. If the trajectory reaches $\vzero$ at some time $t_0$, the process jumps to
\be
Q_{t_0} := \emptyset\in\Q^{(0)}\,.
\ee
The process is required to be a Markov process, so it only remains to specify the jump rate $\sigma_{t_0}(\vartheta_0,\varphi_0) \, \D\vartheta_0 \, \D\varphi_0$ from the 0-particle configuration $\emptyset\in\Q^{(0)}$ to the trajectory in $\Q^{(1)}$ emanating at any given time $t_0$ from $\vzero$ with parameters $\vartheta_0$ and $\varphi_0$. As we will explain, the natural choice analogous to Bell's jump rate formula \cite{Bell86} (and to the jump rates in the non-relativistic case \cite{bohmibc}) is
\be\label{sigmadef}
\sigma_{t_0}(\vartheta_0,\varphi_0) = 
\frac{2(1+q) B}{\pi} \, \frac{\max\bigl\{0,\Im[c^*_-(t_0)\, c_+(t_0)]\bigr\}}{|\psi^{(0)}_{t_0}|^2} {\sin\vartheta_0}\,.
\ee
Here, it is relevant to observe from \eqref{eq:rode} that if $\Im[c^*_- c_+]<0$, then (according to Proposition~\ref{prop:traj}) all trajectories are ingoing, and if $\Im[c^*_- c_+]>0$, then all are outgoing. In the former case, it is not possible to jump onto an outgoing trajectory because there is no outgoing trajectory, and indeed $\sigma_{t_0}=0$. In the latter case, there is a 2-parameter family of outgoing trajectories parameterized by $\vartheta_0$ and $\varphi_0$. The total jump rate (i.e., the rate of leaving $\emptyset$) is
\be\label{totalrate}
\int_0^\pi \D\vartheta_0 \int_0^{2\pi}\D\varphi_0  \, \sigma_{t_0}(\vartheta_0,\varphi_0) = {8(1+q) B} \, \frac{\max\bigl\{0,\Im[c^*_-(t_0)\, c_+(t_0)]\bigr\}}{|\psi^{(0)}_{t_0}|^2} \,.
\ee
Given that a jump occurs at $t_0$, the distribution of the chosen values of $\vartheta_0$ and $\varphi_0$ (i.e., of which trajectory to jump to) has density $(4\pi)^{-1}\sin \vartheta_0$, which means that if we think of $\vartheta_0$ and $\varphi_0$ as coordinates on a sphere, then the distribution is uniform over the sphere. For definiteness, we set that at the time $t_0$ of the jump, $Q_{t_0}:=\emptyset$. This completes the definition of the process.

\subsection{Equivariance and Uniqueness of the Rate}
\label{sec:equivariance}

We now give a non-rigorous justification of the claim that $Q_t$ will be $|\psi_t|^2$ distributed at every $t$. Since in $\Q^{(1)}$ (away from $\vzero$),
\be
\frac{\partial \rho}{\partial t} = -\nabla \cdot \vj\,,
\ee
no $\rho$ is gained or lost there. It follows, first, that away from $\vzero$ probability gets transported by $Q_t$ so as to maintain the density $\rho$ (as usual in Bohmian mechanics \cite{Bohm53,DT09}), and second, that the only place in $\Q^{(1)}$ where $\rho$ is gained or lost is $\vzero$. We now want to express the rate at which $\rho$ is gained or lost there; for simplicity, we write $\psi_{t_0}=\psi$. As before, we neglect how $\psi$ changes near $t_0$. Consider first the flux of probability through the surface element $\D^2\vomega$ of the sphere around $\vzero$ of small but nonzero radius $r$: it is
\be\label{jr}
j_r(r\vomega) \, r^2\, \D^2\vomega\,.
\ee
From Proposition~\ref{prop:current}, we obtain that for small $r$, this is equal to
\be
(C_r + o(r^{1/2-B}))\, \D^2\vomega \,,
\ee 
which for $r\to 0$ converges to $C_r \, d^2\vomega$. Since $C_r=2\pi^{-1}(1+q) B \, \Im[c^*_- c_+]$ is independent of $\vomega$, the rate of gain (positive or negative) of $\rho$ at $\vzero$ is given by $4\pi C_r$. 

This agrees with the rate of gain (positive or negative) of probability at $\vzero$ of $Q_t$: Indeed, if $C_r>0$ then no trajectories end at $\vzero$ at $t_0$ (so no probability is lost there), and the amount transported by jumps from $\emptyset$ to trajectories emanating from $\vzero$ at $t_0$ is the probability at $\emptyset$ times the total jump rate from $\emptyset$, or
\be
|\psi^{(0)}|^2
\int_0^\pi \D\vartheta_0 \int_0^{2\pi}\D\varphi_0  \, \sigma_{t_0}(\vartheta_0,\varphi_0) \stackrel{\eqref{totalrate}}{=} {8(1+q) B} \, \Im[c^*_-\, c_+]=4\pi C_r \,.
\ee
If, however, $C_r<0$ then no upward jumps occur (so no probability is gained at $\vzero$), while the amount lost automatically agrees (since $Q_t$ is $|\psi_t|^2$ distributed) with the flux across the sphere in the limit $r\to 0$.

Finally, to ensure preservation of the $|\psi|^2$ distribution, it remains to verify that the distribution of $Q_t$ over the emanating trajectories agrees with that required for $|\psi_t|^2$, i.e., yields the flux \eqref{jr} through $r\, \D^2\vomega$ in the limit $r\to 0$: Indeed, using that (i)~the leading terms in the radial velocity \eqref{eq:rode} and the azimuthal velocity \eqref{eq:phiode} are independent of $\vomega$, (ii)~the polar velocity \eqref{eq:thetaode} is essentially 0, and (iii)~the distribution defined by $\sigma_{t_0}(\vartheta_0,\varphi_0)$ over the sphere with coordinates $\vartheta_0$ and $\varphi_0$ is uniform as remarked after \eqref{totalrate}, we obtain that the distribution of $Q_t$ over the $r$-sphere is uniform to leading order as $r\to0$. Using again that the leading term in the radial velocity \eqref{eq:rode} is independent of $\vomega$, we obtain that the radial current of $Q_t$ is independent of $\vomega$ in the limit $r\to 0$. Since the total current agrees with $4\pi C_r$, the flux of $Q_t$ through $r\, \D^2\vomega$ agrees with \eqref{jr} in the limit $r\to 0$, as desired. This completes the argument for equivariance.

\bigskip

As a byproduct of this reasoning, we see that conversely, the formula \eqref{sigmadef} is uniquely determined by the demand for equivariance (and Markovianity): Whenever $C_r<0$, $\sigma_{t_0}$ must vanish because there are no outgoing trajectories, and whenever $C_r\geq 0$, $\sigma_{t_0}$ must be given by $C_r |\psi^{(0)}|^{-2} \D^2\vomega$ in order to feed the correct probability distribution into the Bohmian flow.

A further observation is that \eqref{sigmadef} is analogous to the jump rate formula determined in \cite[Sec.~3.1 and 7.2]{bohmibc} for the non-relativistic case; in fact, both formulas can be expressed in a common form if we write $\sigma_{t_0}(r,\vomega)\D^2\vomega$ for the rate, at time $t_0$, for jumping from $\emptyset$ to a trajectory that at radius $r$ will have position in $r\D^2\vomega$:
\be
\lim_{r\to 0}\sigma_{t_0}(r,\vomega) = \lim_{r\to 0} 
\frac{r^2\max\{0,\psi^{(1)}(r\vomega)^\dagger \, \alpha_r \, \psi^{(1)}(r\vomega)\}}{|\psi^{(0)}|^2}
\ee
with $\psi=\psi_{t_0}$. It also becomes evident that the jump rate formula \eqref{sigmadef} is analogous to Bell's jump rate formula \cite{Bell86,DGTZ04}. Presumably, it also arises as a limit of Bell's rate if we can obtain the IBC Hamiltonian as a limit of Hamiltonians with UV cut-off.

\section{Conclusions}
\label{sec:conclusion}

We have studied a model of creation and annihilation of a Dirac particle at a point source at the origin $\vzero$ in $\RRR^3$ and constructed, in a non-rigorous way, a Markov process $(Q_t)_{t\in\RRR}$ in the configuration space $\Q^{(0)}\cup\Q^{(1)}=\{\emptyset\}\cup(\RRR^3\setminus\{\vzero\})$ that is $|\psi_t|^2$ distributed at every time $t$. Since a UV cutoff has the unphysical consequence that a particle can be created at non-zero distance from the source \cite{DGTZ04,bohmibc}, we have used instead an interior-boundary condition (IBC), which has the reasonable consequence that particles can only be created and annihilated directly at the point source. The key element of the definition of the process $(Q_t)_{t \in \RRR}$ was the law \eqref{sigmadef} specifying the creation rate. It is analogous to Bell's jump rate formula \cite{Bell86,DGTZ04}. This process is the first example of a configuration process for a Dirac Hamiltonian with IBC; non-relativistic versions were described in \cite{bohmibc}. We believe that this work might contribute to the extension of Bohmian mechanics to relativistic quantum field theory.

The Hamiltonian $H$ we use was recently constructed rigorously in \cite{HT20} based on prior work in \cite{Hog12,Gal17,GM19}. Some of our considerations here were not rigorous, although all Propositions and Lemmas were proven rigorously. But even the non-rigorous conclusions have benefited from the rigorous construction of $H$; in fact, certain features and details of the process $(Q_t)_{t \in \RRR}$ (such as the fact that a newly created particle circles the $z$ axis infinitely often) have only become accessible due to the detailed information about $H$ (such as the near-$\vzero$ asymptotics of the functions in the domain) provided by its rigorous construction. We have also outlined where we see the biggest hurdle for a full rigorous treatment, and which strategies could be applied to overcome it.

Further questions that would be of interest for future research include whether other models based on Dirac Hamiltonians and IBCs, such as the model of \cite{timelike} in curved space-time, could also be defined rigorously, whether other Dirac Hamiltonians (such as the model of \cite{miki}) would allow for IBCs, what the corresponding Bell-type jump processes look like, and whether there are examples in which the process is qualitatively different from the one described here; in particular, whether there are models for which the jump rate is angle dependent.

\bigskip
\bigskip

\noindent\textit{Conflict of interest.} The authors have no conflicts to disclose. 

\bigskip

\noindent\textit{Acknowledgments.} J.H.~gratefully acknowledges partial financial support by the ERC Advanced Grant ``RMTBeyond" No.~101020331.

\end{document}